\documentclass[draft]{article}
\usepackage{amsfonts,amsmath,amssymb,enumerate,latexsym,times}
\newtheorem{lemma}{Lemma}
\newtheorem{proposition}{Proposition}
\newenvironment{proof}{{\bf Proof:}}{~$\dashv$\\}
\def\N{\mathbb{N}}
\def\KB{\mathbf{KB}}
\def\Int{\mathbf{Int}}
\def\KT{\mathbf{KT}}
\def\KD{\mathbf{KD}}
\def\Alt{\mathbf{Alt}}
\def\K{\mathbf{K}}
\def\GL{\mathbf{GL}}
\def\S{\mathbf{S}}
\def\KTB{\mathbf{KTB}}
\def\KDB{\mathbf{KDB}}
\def\KG{\mathbf{KG}}
\def\KTG{\mathbf{KTG}}
\def\KDG{\mathbf{KDG}}
\begin{document}
\title{About the unification type of
\\
simple symmetric modal logics}
\author{Philippe Balbiani$^{1}$ and \c{C}i\u{g}dem Gencer$^{2}$}
\date{$^{1}$Institut de recherche en informatique de Toulouse
\\
CNRS~---~Toulouse University
\\
Toulouse, France
\\
$^{2}$Faculty of Arts and Sciences
\\
Ayd\i n University
\\
Istanbul, Turkey}
\maketitle
\begin{abstract}
The unification problem in a normal modal logic is to determine, given a formula $\varphi$, whether there exists a substitution $\sigma$ such that $\sigma(\varphi)$ is in that logic.
In that case, $\sigma$ is a unifier of $\varphi$.
We shall say that a set of unifiers of a unifiable formula $\varphi$ is complete if for all unifiers $\sigma$ of $\varphi$, there exists a unifier $\tau$ of $\varphi$ in that set such that $\tau$ is more general than $\sigma$.
When a unifiable formula has no minimal complete set of unifiers, the formula is nullary.
In this paper, we prove that $\KB$, $\KDB$ and $\KTB$ possess nullary formulas.
\end{abstract}
\section{Introduction}
The unification problem in a normal modal logic is to determine, given a formula $\varphi$, whether there exists a substitution $\sigma$ such that $\sigma(\varphi)$ is in that logic.
In that case, $\sigma$ is a unifier of $\varphi$.
We shall say that a set of unifiers of a formula $\varphi$ is complete if for all unifiers $\sigma$ of $\varphi$, there exists a unifier $\tau$ of $\varphi$ in that set such that $\tau$ is more general than $\sigma$.
An important question is the following~\cite{Baader:Ghilardi:2011,Dzik:2007}: when a formula is unifiable, has it a minimal complete set of unifiers?
When the answer is ``no'', the formula is nullary.
When the answer is ``yes'', the formula is unitary, or finitary, or infinitary depending on the cardinalities of its minimal complete sets of unifiers.
A normal modal logic is called nullary if it possesses a nullary formula.
Otherwise, it is called unitary, or finitary, or infinitary depending on the types of its unifiable formulas.
We usually distinguish between elementary unification and unification with parameters.
In elementary unification, all variables are likely to be replaced by formulas when one applies a substitution.
In unification with parameters, some variables~---~called parameters~---~remain unchanged.
\\
\\
It is known that $\mathbf{S}{5}$ is unitary~\cite{Baader:Ghilardi:2011}, $\KT$ is nullary~\cite{Balbiani:to:appear}, $\KD$ is nullary~\cite{Balbiani:Gencer:2017a}, $\Alt_{1}$ is nullary~\cite{Balbiani:Tinchev:2016}, $\S{4.3}$ is unitary~\cite{Dzik:Wojtylak:2012}, transitive normal modal logics like ${\K}4$ are finitary~\cite{Ghilardi:2000} and $\K$ is nullary~\cite{Jerabek:2015}, though the nullariness character of $\KT$ and $\KD$ has only been obtained within the context of unification with parameters.
Taking a look at the literature about unification types in normal modal logics~\cite{Baader:Ghilardi:2011,Dzik:2007}, one will quickly notice that much remains to be done.
For instance, the types of simple Church-Rosser normal modal logics like $\KG$, $\KDG$ and $\KTG$ are unknown\footnote{In this paper, we follow the same conventions as in~\cite{Blackburn:deRijke:Venema:2001,Chagrov:Zakharyaschev:1997,Chellas:1980} for talking about normal modal logics: $\mathbf{S}{5}$ is the least normal modal logic containing the formulas usually denoted $(\mathbf{T})$, $(\mathbf{4})$ and $(\mathbf{B})$, $\KT$ is the least normal modal logic containing the formula usually denoted $(\mathbf{T})$, etc.}.
Even, for all $k\in\N$ such that $k\geq2$, the type of the least normal modal logic containing $\Box^{k}\bot$ is unknown.
In this paper, we adapt to $\KB$, $\KDB$ and $\KTB$ the argument of Je\u{r}\'abek~\cite{Jerabek:2015} showing $\K$ is nullary, though the nullariness character of $\KB$, $\KDB$ and $\KTB$ will only be obtained within the context of unification with parameters.
We assume the reader is at home with tools and techniques in modal logic.
For more on this, see Blackburn {\it et al.}~\cite{Blackburn:deRijke:Venema:2001}, or Chagrov and Zakharyaschev~\cite{Chagrov:Zakharyaschev:1997}, or Chellas~\cite{Chellas:1980}.
\section{Syntax}\label{section:normal:modal:logics:a}
In this section, we present the syntax of normal modal logics.
\paragraph{Formulas}
Let $\mathit{VAR}$ be a nonempty countable set of {\it propositional variables} (with typical members denoted $x$, $y$, etc) and $\mathit{PAR}$ be a nonempty countable set of {\it propositional parameters} (with typical members denoted $p$, $q$, etc).
{\it Atoms} (denoted $\alpha$, $\beta$, etc) are variables or parameters.
The set $\mathit{FOR}$ of all {\it formulas} (with typical members denoted $\varphi$, $\psi$, etc) is inductively defined as follows:
\begin{itemize}
\item $\varphi,\psi::=x\mid p\mid\bot\mid\neg\varphi\mid(\varphi\vee\psi)\mid\Box\varphi$.
\end{itemize}
We adopt the standard rules for omission of the parentheses.
The Boolean connectives $\top$, $\wedge$, $\rightarrow$ and $\leftrightarrow$ are defined by the usual abbreviations.
For all parameters $p$, we write ``$p^{0}$'' to mean ``$\neg p$'' and we write ``$p^{1}$'' to mean ``$p$''.
From now on,
\begin{center}
\begin{tabular}{|c|}
\hline
let $p,q$ be fixed distinct parameters.
\\
\hline
\end{tabular}
\end{center}
Let $\boxplus$ and $\boxminus$ be the modal connectives defined as follows:
\begin{itemize}
\item $\boxplus\varphi::=(p^{0}\wedge q^{0}\rightarrow\Box(p^{1}\wedge q^{0}\rightarrow\Box(p^{0}\wedge q^{1}\rightarrow\Box(p^{0}\wedge q^{0}\rightarrow\varphi))))$,
\item $\boxminus\varphi::=(p^{0}\wedge q^{0}\rightarrow\Box(p^{0}\wedge q^{1}\rightarrow\Box(p^{1}\wedge q^{0}\rightarrow\Box(p^{0}\wedge q^{0}\rightarrow\varphi))))$.
\end{itemize}
For all $k\in\N$, the modal connectives $\boxplus^{k}$ and $\boxminus^{k}$ are inductively defined as follows:
\begin{itemize}
\item $\boxplus^{0}\varphi::=\varphi$,
\item $\boxplus^{k+1}\varphi::=\boxplus\boxplus^{k}\varphi$,
\item $\boxminus^{0}\varphi::=\varphi$,
\item $\boxminus^{k+1}\varphi::=\boxminus\boxminus^{k}\varphi$.
\end{itemize}
For all $k\in\N$, the modal connectives $\boxplus^{<k}$ and $\boxminus^{<k}$ are inductively defined as follows:
\begin{itemize}
\item $\boxplus^{<0}\varphi::=\top$,
\item $\boxplus^{<k+1}\varphi::=(\boxplus^{<k}\varphi\wedge\boxplus^{k}\varphi)$,
\item $\boxminus^{<0}\varphi::=\top$,
\item $\boxminus^{<k+1}\varphi::=(\boxminus^{<k}\varphi\wedge\boxminus^{k}\varphi)$.
\end{itemize}
\paragraph{Degrees}
The {\it degree} of a formula $\varphi$ (in symbols $\deg(\varphi)$) is the nonnegative integer inductively defined as follows:
\begin{itemize}
\item $\deg(x)=0$,
\item $\deg(p)=0$,
\item $\deg(\bot)=0$,
\item $\deg(\neg\varphi)=\deg(\varphi)$,
\item $\deg(\varphi\vee\psi)=\max\{\deg(\varphi),\deg(\psi)\}$,
\item $\deg(\Box\varphi)=\deg(\varphi)+1$.
\end{itemize}
\begin{lemma}
Let $\varphi$ be a formula.
\begin{enumerate}
\item $\deg(\boxplus(\varphi)=\deg(\varphi)+3$,
\item $\deg(\boxminus(\varphi)=\deg(\varphi)+3$,
\item for all $k\in\N$, $\deg(\boxplus^{k}\varphi)=\deg(\varphi)+3k$,
\item for all $k\in\N$, $\deg(\boxminus^{k}\varphi)=\deg(\varphi)+3k$,
\item for all $k\in\N$, if $k=0$ then $\deg(\boxplus^{<k}\varphi)=0$ else $\deg(\boxplus^{<k}\varphi)=\deg(\varphi)+3(k-1)$,
\item for all $k\in\N$, if $k=0$ then $\deg(\boxminus^{<k}\varphi)=0$ else $\deg(\boxminus^{<k}\varphi)=\deg(\varphi)+3(k-1)$.
\end{enumerate}
\end{lemma}
\begin{proof}
$(1)$~and~$(2)$: Left to the reader.
\\
$(3)$--$(6)$: By induction on $k$.
\end{proof}
\paragraph{Substitutions}
A {\it substitution} is a function $\sigma$ associating to each variable $x$ a formula $\sigma(x)$.
Following the standard assumption considered in the literature about the unification problem in normal modal logics~\cite{Baader:Ghilardi:2011,Dzik:2007}, we will always assume that substitutions move at most finitely many variables.
For all formulas $\varphi(x_{1},\ldots,x_{m},p_{1},\ldots,p_{n})$, let $\sigma(\varphi(x_{1},\ldots,x_{m},p_{1},\ldots,p_{n}))$ be $\varphi(\sigma(x_{1}),\ldots,\sigma(x_{m}),p_{1},\ldots,p_{n})$.
The {\it composition} $\sigma\circ\tau$ of the substitutions $\sigma$ and $\tau$ is the substitution associating to each variable $x$ the formula $\tau(\sigma(x))$.
\section{Semantics}\label{section:normal:modal:logics:b}
In this section, we present the semantics of normal modal logics.
\paragraph{Frames and models}
A {\it frame} is a couple $F=(W,R)$ where $W$ is a non-empty set of {\it states} and $R$ is a relation on $W$.
We shall say that a frame $F=(W,R)$ is {\it symmetric} if for all $s,t\in W$, if $sRt$ then $tRs$.
We shall say that a frame $F=(W,R)$ is {\it serial} if for all $s\in W$, there exists $t\in W$ such that $sRt$.
We shall say that a frame $F=(W,R)$ is {\it reflexive} if for all $s\in W$,  $sRs$.
Remark that reflexive frames are serial.
A {\it model} based on a frame $F=(W,R)$ is a triple $M=(W,R,V)$ where $V$ is a function assigning to each variable $x$ a subset $V(x)$ of $W$ and to each parameter $p$ a subset $V(p)$ of $W$.
Given a model $M=(W,R,V)$, the {\it satisfiability} of a modal formula $\varphi$ at $s\in W$ (in symbols $M,s\models\varphi$) is inductively defined as follows:
\begin{itemize}
\item $M,s\models x$ iff $s\in V(x)$,
\item $M,s\models p$ iff $s\in V(p)$,
\item $M,s\not\models\bot$,
\item $M,s\models\neg\varphi$ iff $M,s\not\models\varphi$,
\item $M,s\models\varphi\vee\psi$ iff $M,s\models\varphi$, or $M,s\models\psi$,
\item $M,s\models\Box\varphi$ iff for all $t\in W$, if $sRt$ then $M,t\models\varphi$.
\end{itemize}
\paragraph{Truth and validity}
We shall say that a formula $\varphi$ is {\it true} in a model $M=(W,R,V)$ if $\varphi$ is satisfied at all $s\in W$.
We shall say that a formula $\varphi$ is {\it valid} in a frame $F$ if $\varphi$ is true in all models based on $F$.
We shall say that a formula $\varphi$ is {\it valid} in a class $C$ of frames if $\varphi$ is valid in all frames of $C$.
Let $\KB$ be the set of all formulas valid in the class of all symmetric frames.
Let $\KDB$ be the set of all formulas valid in the class of all serial symmetric frames.
Let $\KTB$ be the set of all formulas valid in the class of all reflexive symmetric frames.
Obviously, $\KB\subseteq\KDB\subseteq\KTB$.
Moreover, $\KB$ is the least normal modal logic containing all formulas of the form $\neg\varphi\rightarrow\Box\neg\Box\varphi$, $\KDB$ is the least normal modal logic containing all formulas of the form $\Box\neg\varphi\rightarrow\neg\Box\varphi$ and $\neg\varphi\rightarrow\Box\neg\Box\varphi$ and $\KTB$ is the least normal modal logic containing all formulas of the form $\Box\varphi\rightarrow\varphi$ and $\neg\varphi\rightarrow\Box\neg\Box\varphi$.
From now on,
\begin{center}
\begin{tabular}{|c|}
\hline
we write ``frame'' to mean ``symmetric frame''.
\\
\hline
\end{tabular}
\end{center}
\begin{lemma}\label{easy:lemma:a}
For all $k\in\N$,
\begin{enumerate}
\item $\boxplus^{k}\top\in\KB$,
\item $\boxminus^{k}\top\in\KB$,
\item $\boxplus^{<k}\top\in\KB$,
\item $\boxminus^{<k}\top\in\KB$.
\end{enumerate}
\end{lemma}
\begin{proof}
By induction on $k$.
\end{proof}
\begin{lemma}\label{easy:lemma:b}
For all $k\in\N$,
\begin{enumerate}
\item $\boxplus^{k}\bot\not\in\KB$,
\item $\boxminus^{k}\bot\not\in\KB$.
\end{enumerate}
\end{lemma}
\begin{proof}
Let $k\in\N$.
Let $F=(W,R)$ where $W=\{0,\ldots,3k\}$ and $R=\{(i,j):\ {\mid}j-i{\mid}\leq1\}$.
Let $M=(W,R,V)$ where $V(p)=\{i:\ i=1\mod3\}$, $V(q)=\{i:\ i=2\mod3\}$ and for all atoms $\alpha$, if $\alpha\not=p$ and $\alpha\not=q$ then $V(\alpha)=\emptyset$.
The reader may easily verify that $M,0\not\models\boxplus^{k}\bot$ and $M,3k\not\models\boxminus^{k}\bot$.
Hence, $\boxplus^{k}\bot\not\in\KB$ and $\boxminus^{k}\bot\not\in\KB$.
\end{proof}
In the proof of Lemma~\ref{easy:lemma:b}, remark that the frame $F=(W,R)$ is reflexive.
\begin{lemma}\label{lemma:about:box:less:than}
Let $\varphi$ be a formula.
For all $k\in\N$,
\begin{enumerate}
\item $(\boxplus^{<k+1}\varphi\leftrightarrow\varphi\wedge\boxplus\boxplus^{<k}\varphi)\in\KB$,
\item $(\boxminus^{<k+1}\varphi\leftrightarrow\varphi\wedge\boxminus\boxminus^{<k}\varphi)\in\KB$.
\end{enumerate}
\end{lemma}
\begin{proof}
By induction on $k$.
\end{proof}
\begin{lemma}\label{lemma:about:k:l:and:boxes}
For all $k,l\in\N$,
\begin{enumerate}
\item if $k>l$ then $(\boxplus^{k}\bot\rightarrow\boxplus^{l}\bot)\not\in\KB$,
\item if $k>l$ then $(\boxminus^{k}\bot\rightarrow\boxminus^{l}\bot)\not\in\KB$.
\end{enumerate}
\end{lemma}
\begin{proof}
Let $k\in\N$.
Suppose $k>l$.
Let $F=(W,R)$ where $W=\{0,\ldots,3l\}$ and $R=\{(i,j):\ {\mid}j-i{\mid}\leq1\}$.
Let $M=(W,R,V)$ where $V(p)=\{i:\ i=1\mod3\}$, $V(q)=\{i:\ i=2\mod3\}$ and for all atoms $\alpha$, if $\alpha\not=p$ and $\alpha\not=q$ then $V(\alpha)=\emptyset$.
The reader may easily verify that $M,0\not\models(\boxplus^{k}\bot\rightarrow\boxplus^{l}\bot)$ and $M,3l\not\models(\boxminus^{k}\bot\rightarrow\boxminus^{l}\bot)$.
Hence, $(\boxplus^{k}\bot\rightarrow\boxplus^{l}\bot)\not\in\KB$ and $(\boxminus^{k}\bot\rightarrow\boxminus^{l}\bot)\not\in\KB$.
\end{proof}
In the proof of Lemma~\ref{lemma:about:k:l:and:boxes}, remark that the frame $F=(W,R)$ is reflexive.
\begin{lemma}\label{proposition:tense:modalities}
For all formulas $\varphi$, $(\varphi\rightarrow\boxplus\varphi)\in\KB$ iff $(\neg\varphi\rightarrow\boxminus\neg\varphi)\in\KB$.
\end{lemma}
\begin{proof}
Let $\varphi$ be a formula such that $(\varphi\rightarrow\boxplus\varphi)\not\in\KB$ and $(\neg\varphi\rightarrow\boxminus\neg\varphi)\in\KB$, or $(\varphi\rightarrow\boxplus\varphi)\in\KB$ and $(\neg\varphi\rightarrow\boxminus\neg\varphi)\not\in\KB$.
\\
{\bf ---~Case ``$(\varphi\rightarrow\boxplus\varphi)\not\in\KB$ and $(\neg\varphi\rightarrow\boxminus\neg\varphi)\in\KB$'':}
Let $F=(W,R)$ be a frame, $M=(W,R,V)$ be a model based on $F$ and $s\in W$ be such that $M,s\not\models(\varphi\rightarrow\boxplus\varphi)$.
Hence, $M,s\models\varphi$ and $M,s\not\models\boxplus\varphi$.
Let $t,u,v\in W$ be such that $sRt$, $tRu$, $uRv$, $M,s\models p^{0}\wedge q^{0}$, $M,t\models p^{1}\wedge q^{0}$, $M,u\models p^{0}\wedge q^{1}$, $M,v\models p^{0}\wedge q^{0}$ and $M,v\not\models\varphi$.
Thus, $tRs$, $uRt$ and $vRu$.
Moreover, $M,v\models\neg\varphi$.
Since $(\neg\varphi\rightarrow\boxminus\neg\varphi)\in\KB$, therefore $M,v\models(\neg\varphi\rightarrow\boxminus\neg\varphi)\in\KB$.
Since $M,v\models\neg\varphi$, therefore $M,v\models\boxminus\neg\varphi$.
Since $tRs$, $uRt$, $vRu$, $M,s\models p^{0}\wedge q^{0}$, $M,t\models p^{1}\wedge q^{0}$, $M,u\models p^{0}\wedge q^{1}$ and $M,v\models p^{0}\wedge q^{0}$, therefore $M,s\models\neg\varphi$.
Consequently, $M,s\not\models\varphi$: a contradiction.
\\
{\bf ---~Case ``$(\varphi\rightarrow\boxplus\varphi)\in\KB$ and $(\neg\varphi\rightarrow\boxminus\neg\varphi)\not\in\KB$'':}
Let $F=(W,R)$ be a frame, $M=(W,R,V)$ be a model based on $F$ and $s\in W$ be such that $M,s\not\models(\neg\varphi\rightarrow\boxminus\neg\varphi)$.
Hence, $M,s\models\neg\varphi$ and $M,s\not\models\boxminus\neg\varphi$.
Let $t,u,v\in W$ be such that $sRt$, $tRu$, $uRv$, $M,s\models p^{0}\wedge q^{0}$, $M,t\models p^{0}\wedge q^{1}$, $M,u\models p^{1}\wedge q^{0}$, $M,v\models p^{0}\wedge q^{0}$ and $M,v\not\models\neg\varphi$.
Thus, $tRs$, $uRt$ and $vRu$.
Moreover, $M,v\models\varphi$.
Since $(\varphi\rightarrow\boxplus\varphi)\in\KB$, therefore $M,v\models(\varphi\rightarrow\boxplus\varphi)\in\KB$.
Since $M,v\models\varphi$, therefore $M,v\models\boxplus\varphi$.
Since $tRs$, $uRt$, $vRu$, $M,s\models p^{0}\wedge q^{0}$, $M,t\models p^{0}\wedge q^{1}$, $M,u\models p^{1}\wedge q^{0}$ and $M,v\models p^{0}\wedge q^{0}$, therefore $M,s\models\varphi$.
Consequently, $M,s\not\models\neg\varphi$: a contradiction.
\end{proof}
\section{Unification}\label{section:unification}
In this section, we present unification in $\KB$.
\paragraph{Unification problem}
We shall say that a substitution $\sigma$ is {\it equivalent} to a substitution $\tau$ (in symbols $\sigma\simeq\tau$) if for all variables $x$, $(\sigma(x)\leftrightarrow\tau(x))\in\KB$.
We shall say that a substitution $\sigma$ is more {\it general} than a substitution $\tau$ (in symbols $\sigma\preceq\tau$) if there exists a substitution $\upsilon$ such that $\sigma\circ\upsilon\simeq\tau$.
Obviously, $\preceq$ contains $\simeq$.
Moreover,
\begin{proposition}[Baader and Ghilardi~\cite{Baader:Ghilardi:2011}, Dzik~\cite{Dzik:2007}]\label{lemma:simeq:ref:sym:tra}
\begin{enumerate}
\item The binary relation $\simeq$ is reflexive, symmetric and transitive on the set of all substitutions,
\item the binary relation $\preceq$ is reflexive and transitive on the set of all substitutions.
\end{enumerate}
\end{proposition}
We shall say that a set $\Sigma$ of substitutions is {\it minimal} if for all $\sigma,\tau\in\Sigma$, if $\sigma\preceq\tau$ then $\sigma\simeq\tau$.
We shall say that a formula $\varphi$ is {\it unifiable} if there exists a substitution $\sigma$ such that $\sigma(\varphi)\in\KB$.
In that case, $\sigma$ is a {\it unifier} of $\varphi$.
\begin{proposition}\label{normal:unifiers:are:enough}
Let $\varphi$ be a formula.
For all unifiers $\sigma$ of $\varphi$, there exists a unifier $\tau$ of $\varphi$ such that $\tau\preceq\sigma$ and for all variables $x$, if $x$ does not occur in $\varphi$ then $\tau(x)=x$.
\end{proposition}
\begin{proof}
Left to the reader.
\end{proof}
We shall say that a set $\Sigma$ of unifiers of a unifiable formula $\varphi$ is {\it complete} if for all unifiers $\sigma$ of $\varphi$, there exists $\tau\in\Sigma$ such that $\tau\preceq\sigma$.
\paragraph{Unification types}
An important question is the following: when a formula is unifiable, has it a minimal complete set of unifiers?
When the answer is ``yes'', how large is this set?
We shall say that a unifiable formula
\begin{itemize}
\item $\varphi$ is {\it nullary} if there exists no minimal complete set of unifiers of $\varphi$,
\item $\varphi$ is {\it unitary} if there exists a minimal complete set of unifiers of $\varphi$ with cardinality $1$,
\item $\varphi$ is {\it finitary} if there exists a finite minimal complete set of unifiers of $\varphi$ but there exists no with cardinality $1$,
\item $\varphi$ is {\it infinitary} if there exists a minimal complete set of unifiers of $\varphi$ but there exists no finite one.
\end{itemize}
\section{Playing with substitutions}\label{section:playing}
For all $k\in\N$, let $\sigma_{k}$ and $\tau_{k}$ be the substitutions inductively defined as follows:
\begin{itemize}
\item $\sigma_{0}(x)=\bot$,
\item for all variables $y$ distinct from $x$, $\sigma_{0}(y)=y$,
\item $\tau_{0}(x)=\top$,
\item for all variables $y$ distinct from $x$, $\tau_{0}(y)=y$,
\item $\sigma_{k+1}(x)=(x\wedge\boxplus\sigma_{k}(x))$,
\item for all variables $y$ distinct from $x$, $\sigma_{k+1}(y)=y$,
\item $\tau_{k+1}(x)=\neg(\neg x\wedge\boxminus\neg\tau_{k}(x))$,
\item for all variables $y$ distinct from $x$, $\tau_{k+1}(y)=y$.
\end{itemize}
These substitutions will be used in Section~\ref{section:KB:is:nullary} to prove that $\KB$ possesses nullary formulas.
\begin{lemma}\label{lemma:to:be:used:later}
For all $k\in\N$,
\begin{enumerate}
\item $(\boxplus^{<k}x\wedge\boxplus^{k}\bot\rightarrow\sigma_{k}(x))\in\KB$,
\item $(\boxminus^{<k}\neg x\wedge\boxminus^{k}\bot\rightarrow\neg\tau_{k}(x))\in\KB$.
\end{enumerate}
\end{lemma}
\begin{proof}
By induction on $k$.
\end{proof}
\begin{lemma}\label{lemma:sigma:tau:imply:x}
For all $k\in\N$,
\begin{enumerate}
\item $(\sigma_{k}(x)\rightarrow x)\in\KB$,
\item $(\neg\tau_{k}(x)\rightarrow\neg x)\in\KB$.
\end{enumerate}
\end{lemma}
\begin{proof}
By induction on $k$.
\end{proof}
\begin{lemma}\label{lemma:sigma:tau:imply:box:x}
For all $k\in\N$,
\begin{enumerate}
\item $(\sigma_{k}(x)\rightarrow\boxplus\sigma_{k}(x))\in\KB$,
\item $(\neg\tau_{k}(x)\rightarrow\boxminus\neg\tau_{k}(x))\in\KB$.
\end{enumerate}
\end{lemma}
\begin{proof}
By induction on $k$.
\end{proof}
\begin{lemma}\label{lemma:sigma:tau:imply:box:bot:bot}
For all $k,l\in\N$,
\begin{enumerate}
\item if $k\leq l$ then $(\sigma_{k}(x)\rightarrow\boxplus^{l}\bot)\in\KB$,
\item if $k\leq l$ then $(\neg\tau_{k}(x)\rightarrow\boxminus^{l}\bot)\in\KB$.
\end{enumerate}
\end{lemma}
\begin{proof}
By induction on $k$.
\end{proof}
\begin{lemma}\label{lemma:sigma:tau:imply:box:bot:bot:k:greater:than:l}
For all $k,l\in\N$,
\begin{enumerate}
\item if $k>l$ then $(\sigma_{k}(x)\rightarrow\boxplus^{l}\bot)\not\in\KB$,
\item if $k>l$ then $(\neg\tau_{k}(x)\rightarrow\boxminus^{l}\bot)\not\in\KB$.
\end{enumerate}
\end{lemma}
\begin{proof}
Let $k,l\in\N$.
\\
$(1)$: Suppose $k>l$ and $(\sigma_{k}(x)\rightarrow\boxplus^{l}\bot)\in\KB$.
Let $\upsilon$ be the substitution defined as follows:
\begin{itemize}
\item $\upsilon(x)=\top$,
\item for all variables $y$ distinct from $x$, $\upsilon(y)=y$.
\end{itemize}
Since $(\sigma_{k}(x)\rightarrow\boxplus^{l}\bot)\in\KB$, therefore $(\upsilon(\sigma_{k}(x))\rightarrow\boxplus^{l}\bot)\in\KB$.
By Lemma~\ref{lemma:to:be:used:later}, $(\boxplus^{<k}x\wedge\boxplus^{k}\bot\rightarrow\sigma_{k}(x))\in\KB$.
Hence, $(\boxplus^{<k}\upsilon(x)\wedge\boxplus^{k}\bot\rightarrow\upsilon(\sigma_{k}(x)))\in\KB$.
Since $\upsilon(x)=\top$, therefore by Lemma~\ref{easy:lemma:a}, $(\boxplus^{k}\bot\rightarrow\upsilon(\sigma_{k}(x)))\in\KB$.
Since $(\upsilon(\sigma_{k}(x))\rightarrow\boxplus^{l}\bot)\in\KB$, therefore $(\boxplus^{k}\bot\rightarrow\boxplus^{l}\bot)\in\KB$.
Thus, by Lemma~\ref{lemma:about:k:l:and:boxes}, $k\not>l$: a contradiction.
\\
$(2)$: Suppose $k>l$ and $(\neg\tau_{k}(x)\rightarrow\boxminus^{l}\bot)\in\KB$.
Let $\upsilon$ be the substitution defined as follows:
\begin{itemize}
\item $\upsilon(x)=\bot$,
\item for all variables $y$ distinct from $x$, $\upsilon(y)=y$.
\end{itemize}
Since $(\neg\tau_{k}(x)\rightarrow\boxminus^{l}\bot)\in\KB$, therefore $(\upsilon(\neg\tau_{k}(x))\rightarrow\boxminus^{l}\bot)\in\KB$.
By Lemma~\ref{lemma:to:be:used:later}, $(\boxminus^{<k}\neg x\wedge\boxminus^{k}\bot\rightarrow\neg\tau_{k}(x))\in\KB$.
Hence, $(\boxminus^{<k}\neg\upsilon(x)\wedge\boxminus^{k}\bot\rightarrow\upsilon(\neg\tau_{k}(x)))\in\KB$.
Since $\upsilon(x)=\bot$, therefore by Lemma~\ref{easy:lemma:a}, $(\boxminus^{k}\bot\rightarrow\upsilon(\neg\tau_{k}(x)))\in\KB$.
Since $(\upsilon(\neg\tau_{k}(x))\rightarrow\boxminus^{l}\bot)\in\KB$, therefore $(\boxminus^{k}\bot\rightarrow\boxminus^{l}\bot)\in\KB$.
Thus, by Lemma~\ref{lemma:about:k:l:and:boxes}, $k\not>l$: a contradiction.
\end{proof}
\begin{lemma}\label{lemma:sigma:tau:imply:not:the:case:this:time}
For all $k,l\in\N$,
\begin{enumerate}
\item $(\boxplus^{k}\bot\vee\neg\tau_{l}(x))\not\in\KB$,
\item $(\boxminus^{k}\bot\vee\sigma_{l}(x))\not\in\KB$.
\end{enumerate}
\end{lemma}
\begin{proof}
Let $k,l\in\N$.
\\
$(1)$: Suppose $(\boxplus^{k}\bot\vee\neg\tau_{l}(x))\in\KB$.
By Lemma~\ref{lemma:sigma:tau:imply:x}, $(\neg\tau_{l}(x)\rightarrow\neg x)\in\KB$.
Since $(\boxplus^{k}\bot\vee\neg\tau_{l}(x))\in\KB$, therefore $(\boxplus^{k}\bot\vee\neg x)\in\KB$.
Let $\upsilon$ be the substitution defined as follows:
\begin{itemize}
\item $\upsilon(x)=\top$,
\item for all variables $y$ distinct from $x$, $\upsilon(y)=y$.
\end{itemize}
Since $(\boxplus^{k}\bot\vee\neg x)\in\KB$, therefore $(\boxplus^{k}\bot\vee\neg\upsilon(x))\in\KB$.
Since $\upsilon(x)=\top$, therefore $\boxplus^{k}\bot\in\KB$: a contradiction with Lemma~\ref{easy:lemma:b}.
\\
$(2)$: Suppose $(\boxminus^{k}\bot\vee\sigma_{l}(x))\in\KB$.
By Lemma~\ref{lemma:sigma:tau:imply:x}, $(\sigma_{l}(x)\rightarrow x)\in\KB$.
Since $(\boxminus^{k}\bot\vee\sigma_{l}(x))\in\KB$, therefore $(\boxminus^{k}\bot\vee x)\in\KB$.
Let $\upsilon$ be the substitution defined as follows:
\begin{itemize}
\item $\upsilon(x)=\bot$,
\item for all variables $y$ distinct from $x$, $\upsilon(y)=y$.
\end{itemize}
Since $(\boxminus^{k}\bot\vee x)\in\KB$, therefore $(\boxminus^{k}\bot\vee\upsilon(x))\in\KB$.
Since $\upsilon(x)=\bot$, therefore $\boxminus^{k}\bot\in\KB$: a contradiction with Lemma~\ref{easy:lemma:b}.
\end{proof}
\begin{lemma}\label{lemma:sigma:lambda:k:l:and:also:tau:mu:k:l:pre:1}
For all $k,l\in\N$,
\begin{enumerate}
\item if $k\leq l$ then $(\boxplus^{k}\bot\wedge\sigma_{l}(x)\leftrightarrow\sigma_{k}(x))$,
\item if $k\leq l$ then $(\boxminus^{k}\bot\wedge\neg\tau_{l}(x)\leftrightarrow\neg\tau_{k}(x))$,
\end{enumerate}
\end{lemma}
\begin{proof}
By induction on $k$.
\end{proof}
For all $k\in\N$, let $\lambda_{k}$ and $\mu_{k}$ be the substitutions defined as follows:
\begin{itemize}
\item $\lambda_{k}(x)=(x\wedge\boxplus^{k}\bot)$,
\item for all variables $y$ distinct from $x$, $\lambda_{k}(y)=y$,
\item $\mu_{k}(x)=\neg(\neg x\wedge\boxminus^{k}\bot)$,
\item for all variables $y$ distinct from $x$, $\mu_{k}(y)=y$.
\end{itemize}
\begin{lemma}\label{lemma:sigma:lambda:k:l:and:also:tau:mu:k:l:pre:2}
For all $k,l\in\N$,
\begin{enumerate}
\item if $k\leq l$ then $(\lambda_{l}(\sigma_{k}(x))\leftrightarrow\sigma_{k}(x))\in\KB$,
\item if $k\leq l$ then $(\mu_{l}(\tau_{k}(x))\leftrightarrow\tau_{k}(x))\in\KB$.
\end{enumerate}
\end{lemma}
\begin{proof}
By induction on $k$.
\end{proof}
\begin{lemma}\label{lemma:sigma:lambda:k:l:and:also:tau:mu:k:l:pre:3}
For all $k,l\in\N$,
\begin{enumerate}
\item if $k\geq l$ then $(\lambda_{l}(\sigma_{k}(x))\leftrightarrow\sigma_{l}(x))\in\KB$,
\item if $k\geq l$ then $(\mu_{l}(\tau_{k}(x))\leftrightarrow\tau_{l}(x))\in\KB$.
\end{enumerate}
\end{lemma}
\begin{proof}
By induction on $k$.
\end{proof}
\begin{lemma}\label{lemma:sigma:lambda:k:l:and:also:tau:mu:k:l}
For all $k,l\in\N$,
\begin{enumerate}
\item if $k\leq l$ then $\sigma_{l}\circ\lambda_{k}\simeq\sigma_{k}$,
\item if $k\leq l$ then $\tau_{l}\circ\mu_{k}\simeq\tau_{k}$.
\end{enumerate}
\end{lemma}
\begin{proof}
By Lemma~\ref{lemma:sigma:lambda:k:l:and:also:tau:mu:k:l:pre:3}.
\end{proof}
\begin{lemma}\label{lemma:0:K:q}
For all $k,l\in\N$,
\begin{enumerate}
\item if $k\leq l$ then $\sigma_{l}\preceq\sigma_{k}$,
\item if $k\leq l$ then $\tau_{l}\preceq\tau_{k}$.
\end{enumerate}
\end{lemma}
\begin{proof}
By Lemma~\ref{lemma:sigma:lambda:k:l:and:also:tau:mu:k:l}.
\end{proof}
\begin{lemma}\label{lemma:0:K:r}
For all $k,l\in\N$,
\begin{enumerate}
\item if $k<l$ then $\sigma_{k}\not\preceq\sigma_{l}$,
\item if $k<l$ then $\tau_{k}\not\preceq\tau_{l}$.
\end{enumerate}
\end{lemma}
\begin{proof}
Let $k,l\in\N$.
\\
$(1)$: Suppose $k<l$ and $\sigma_{k}\preceq\sigma_{l}$.
Let $\lambda$ be a substitution such that $\sigma_{k}\circ\lambda\simeq\sigma_{l}$.
Hence, $(\lambda(\sigma_{k}(x))\leftrightarrow\sigma_{l}(x))\in\KB$.
By Lemma~\ref{lemma:sigma:tau:imply:box:bot:bot}, $(\sigma_{k}(x)\rightarrow\boxplus^{k}\bot)\in\KB$.
Thus, $(\lambda(\sigma_{k}(x))\rightarrow\boxplus^{k}\bot)\in\KB$.
Since $(\lambda(\sigma_{k}(x))\leftrightarrow\sigma_{l}(x))\in\KB$, therefore $(\sigma_{l}(x)\rightarrow\boxplus^{k}\bot)\in\KB$.
Consequently, by Lemma~\ref{lemma:sigma:tau:imply:box:bot:bot:k:greater:than:l}, $l\not>k$: a contradiction.
\\
$(2)$: Suppose $k<l$ and $\tau_{k}\preceq\tau_{l}$.
Let $\mu$ be a substitution such that $\tau_{k}\circ\mu\simeq\tau_{l}$.
Hence, $(\mu(\tau_{k}(x))\leftrightarrow\tau_{l}(x))\in\KB$.
By Lemma~\ref{lemma:sigma:tau:imply:box:bot:bot}, $(\neg\tau_{k}(x)\rightarrow\boxminus^{k}\bot)\in\KB$.
Thus, $(\mu(\neg\tau_{k}(x))\rightarrow\boxminus^{k}\bot)\in\KB$.
Since $(\mu(\tau_{k}(x))\leftrightarrow\tau_{l}(x))\in\KB$, therefore $(\neg\tau_{l}(x)\rightarrow\boxminus^{k}\bot)\in\KB$.
Consequently, by Lemma~\ref{lemma:sigma:tau:imply:box:bot:bot:k:greater:than:l}, $l\not>k$: a contradiction.
\end{proof}
\begin{lemma}\label{lemma:0:K:qr}
For all $k,l\in\N$,
\begin{enumerate}
\item $\sigma_{k}\not\preceq\tau_{l}$,
\item $\tau_{k}\not\preceq\sigma_{l}$.
\end{enumerate}
\end{lemma}
\begin{proof}
Let $k,l\in\N$.
\\
$(1)$: Suppose $\sigma_{k}\preceq\tau_{l}$.
Let $\upsilon$ be a substitution such that $\sigma_{k}\circ\upsilon\simeq\tau_{l}$.
Hence, $(\upsilon(\sigma_{k}(x))\leftrightarrow\tau_{l}(x))\in\KB$.
By Lemma~\ref{lemma:sigma:tau:imply:box:bot:bot}, $(\sigma_{k}(x)\rightarrow\boxplus^{k}\bot)\in\KB$.
Thus, $(\upsilon(\sigma_{k}(x))\rightarrow\boxplus^{k}\bot)\in\KB$.
Since $(\upsilon(\sigma_{k}(x))\leftrightarrow\tau_{l}(x))\in\KB$, therefore $(\boxplus^{k}\bot\vee\neg\tau_{l}(x))\in\KB$: a contradiction with Lemma~\ref{lemma:sigma:tau:imply:not:the:case:this:time}.
\\
$(2)$: Suppose $\tau_{k}\preceq\sigma_{l}$.
Let $\upsilon$ be a substitution such that $\tau_{k}\circ\upsilon\simeq\sigma_{l}$.
Hence, $(\upsilon(\tau_{k}(x))\leftrightarrow\sigma_{l}(x))\in\KB$.
By Lemma~\ref{lemma:sigma:tau:imply:box:bot:bot}, $(\neg\tau_{k}(x)\rightarrow\boxminus^{k}\bot)\in\KB$.
Thus, $(\upsilon(\neg\tau_{k}(x))\rightarrow\boxminus^{k}\bot)\in\KB$.
Since $(\upsilon(\tau_{k}(x))\leftrightarrow\sigma_{l}(x))\in\KB$, therefore $(\boxminus^{k}\bot\vee\sigma_{l}(x))\in\KB$: a contradiction with Lemma~\ref{lemma:sigma:tau:imply:not:the:case:this:time}.
\end{proof}
\section{About the nullariness of $\KB$}\label{section:KB:is:nullary}
In this section, we prove that the following formula is unifiable and nullary:
\begin{itemize}
\item $\varphi::=((x\rightarrow\boxplus x)\wedge(\neg x\rightarrow\boxminus\neg x))$.
\end{itemize}
\begin{lemma}\label{lemma:every:unifier:of:varphi:has:this:property:1}
Let $\sigma$ be a unifier of $\varphi$.
For all $k\in\N$,
\begin{enumerate}
\item $(\sigma(x)\rightarrow\boxplus^{<k}\sigma(x))\in\KB$,
\item $(\neg\sigma(x)\rightarrow\boxminus^{<k}\neg\sigma(x))\in\KB$.
\end{enumerate}
\end{lemma}
\begin{proof}
By induction on $k$.
\end{proof}
\begin{lemma}\label{lemma:0:K}
For all $k\in\N$,
\begin{enumerate}
\item $\sigma_{k}$ is a unifier of $\varphi$,
\item $\tau_{k}$ is a unifier of $\varphi$.
\end{enumerate}
\end{lemma}
\begin{proof}
By Lemmas~\ref{proposition:tense:modalities} and~\ref{lemma:sigma:tau:imply:box:x}.
\end{proof}
\begin{lemma}\label{lemma:4:K}
Let $\upsilon$ be a substitution.
If $\upsilon$ is a unifier of $\varphi$ then
\begin{enumerate}
\item for all $k\in\N$, the following conditions are equivalent: $(a)$~$\sigma_{k}\circ\upsilon\simeq\upsilon$, $(b)$~$\sigma_{k}\preceq\upsilon$, $(c)$~$(\upsilon(x)\rightarrow\boxplus^{k}\bot)\in\KB$,
\item for all $k\in\N$, the following conditions are equivalent: $(d)$~$\tau_{k}\circ\upsilon\simeq\upsilon$, $(e)$~$\tau_{k}\preceq\upsilon$, $(f)$~$(\neg\upsilon(x)\rightarrow\boxminus^{k}\bot)\in\KB$.
\end{enumerate}
\end{lemma}
\begin{proof}
Suppose $\upsilon$ is a unifier of $\varphi$.
\\
$(1)$: Let $k\in\N$.
\\
$(a)\Rightarrow(b)$: Suppose $\sigma_{k}\circ\upsilon\simeq\upsilon$.
Hence, $\sigma_{k}\preceq\upsilon$.
\\
$(b)\Rightarrow(c)$: Suppose $\sigma_{k}\preceq\upsilon$.
Let $\upsilon^{\prime}$ be a substitution such that $\sigma_{k}\circ\upsilon^{\prime}\simeq\upsilon$.
Hence, $(\upsilon^{\prime}(\sigma_{k}(x))\leftrightarrow\upsilon(x))\in\KB$.
By Lemma~\ref{lemma:sigma:tau:imply:box:bot:bot}, $(\sigma_{k}(x)\rightarrow\boxplus^{k}\bot)\in\KB$.
Thus, $(\upsilon^{\prime}(\sigma_{k}(x))\rightarrow\boxplus^{k}\bot)\in\KB$.
Since $(\upsilon^{\prime}(\sigma_{k}(x))\leftrightarrow\upsilon(x))\in\KB$, therefore $(\upsilon(x)\rightarrow\boxplus^{k}\bot)\in\KB$.
\\
$(c)\Rightarrow(a)$: Suppose $(\upsilon(x)\rightarrow\boxplus^{k}\bot)\in\KB$.
Since $\upsilon$ is a unifier of $\varphi$, therefore by Lemma~\ref{lemma:every:unifier:of:varphi:has:this:property:1}, $(\upsilon(x)\rightarrow\boxplus^{<k}\upsilon(x))\in\KB$.
Since $(\upsilon(x)\rightarrow\boxplus^{k}\bot)\in\KB$, therefore $(\upsilon(x)\rightarrow\boxplus^{<k}\upsilon(x)\wedge\boxplus^{k}\bot)\in\KB$.
By Lemma~\ref{lemma:to:be:used:later}, $(\boxplus^{<k}x\wedge\boxplus^{k}\bot\rightarrow\sigma_{k}(x))\in\KB$.
Hence, $(\boxplus^{<k}\upsilon(x)\wedge\boxplus^{k}\bot\rightarrow\upsilon(\sigma_{k}(x)))\in\KB$.
Since $(\upsilon(x)\rightarrow\boxplus^{<k}\upsilon(x)\wedge\boxplus^{k}\bot)\in\KB$, therefore $(\upsilon(x)\rightarrow\upsilon(\sigma_{k}(x)))\in\KB$.
By Lemma~\ref{lemma:sigma:tau:imply:x}, $(\sigma_{k}(x)\rightarrow x)\in\KB$.
Thus, $(\upsilon(\sigma_{k}(x))\rightarrow\upsilon(x))\in\KB$.
Since $(\upsilon(x)\rightarrow\upsilon(\sigma_{k}(x)))\in\KB$, therefore $(\upsilon(\sigma_{k}(x))\leftrightarrow\upsilon(x))\in\KB$.
Consequently, $\sigma_{k}\circ\upsilon\simeq\upsilon$.
\\
$(2)$: Let $k\in\N$.
\\
$(d)\Rightarrow(e)$: Suppose $\tau_{k}\circ\upsilon\simeq\upsilon$.
Hence, $\tau_{k}\preceq\upsilon$.
\\
$(e)\Rightarrow(f)$: Suppose $\tau_{k}\preceq\upsilon$.
Let $\upsilon^{\prime}$ be a substitution such that $\tau_{k}\circ\upsilon^{\prime}\simeq\upsilon$.
Hence, $(\upsilon^{\prime}(\tau_{k}(x))\leftrightarrow\upsilon(x))\in\KB$.
By Lemma~\ref{lemma:sigma:tau:imply:box:bot:bot}, $(\neg\tau_{k}(x)\rightarrow\boxminus^{k}\bot)\in\KB$.
Thus, $(\upsilon^{\prime}(\neg\tau_{k}(x))\rightarrow\boxminus^{k}\bot)\in\KB$.
Since $(\upsilon^{\prime}(\tau_{k}(x))\leftrightarrow\upsilon(x))\in\KB$, therefore $(\neg\upsilon(x)\rightarrow\boxminus^{k}\bot)\in\KB$.
\\
$(f)\Rightarrow(d)$: Suppose $(\neg\upsilon(x)\rightarrow\boxminus^{k}\bot)\in\KB$.
Since $\upsilon$ is a unifier of $\varphi$, therefore by Lemma~\ref{lemma:every:unifier:of:varphi:has:this:property:1}, $(\neg\upsilon(x)\rightarrow\boxminus^{<k}\neg\upsilon(x))\in\KB$.
Since $(\neg\upsilon(x)\rightarrow\boxminus^{k}\bot)\in\KB$, therefore $(\neg\upsilon(x)\rightarrow\boxminus^{<k}\neg\upsilon(x)\wedge\boxminus^{k}\bot)\in\KB$.
By Lemma~\ref{lemma:to:be:used:later}, $(\boxminus^{<k}\neg x\wedge\boxminus^{k}\bot\rightarrow\neg\tau_{k}(x))\in\KB$.
Hence, $(\boxminus^{<k}\neg\upsilon(x)\wedge\boxminus^{k}\bot\rightarrow\upsilon(\neg\tau{k}(x)))\in\KB$.
Since $(\neg\upsilon(x)\rightarrow\boxminus^{<k}\neg\upsilon(x)\wedge\boxminus^{k}\bot)\in\KB$, therefore $(\neg\upsilon(x)\rightarrow\upsilon(\neg\tau_{k}(x)))\in\KB$.
By Lemma~\ref{lemma:sigma:tau:imply:x}, $(\neg\tau_{k}(x)\rightarrow\neg x)\in\KB$.
Thus, $(\upsilon(\neg\tau_{k}(x))\rightarrow\neg\upsilon(x))\in\KB$.
Since $(\neg\upsilon(x)\rightarrow\upsilon(\neg\tau_{k}(x)))\in\KB$, therefore $(\upsilon(\tau_{k}(x))\leftrightarrow\upsilon(x))\in\KB$.
Consequently, $\tau_{k}\circ\upsilon\simeq\upsilon$.
\end{proof}
\begin{lemma}\label{lemma:6:K}
Let $\sigma$ be a substitution.
If $\sigma$ is a unifier of $\varphi$ then there exists $k\in\N$ such that $\sigma_{k}\preceq\sigma$, or $\tau_{k}\preceq\sigma$.
\end{lemma}
\begin{proof}
Suppose $\sigma$ is a unifier of $\varphi$.
By Propositions~\ref{lemma:simeq:ref:sym:tra} and~\ref{normal:unifiers:are:enough}, we can assume that for all variables $y$ distinct from $x$, $\sigma(y)=y$.
Let $k\in\N$ be such that $\deg(\sigma(x))\leq3k$.
Suppose $\sigma_{k}\not\preceq\sigma$ and $\tau_{k}\not\preceq\sigma$.
Since $\sigma$ is a unifier of $\varphi$, therefore by Lemma~\ref{lemma:4:K}, $(\sigma(x)\rightarrow\boxplus^{k}\bot)\not\in\KB$ and $(\neg\sigma(x)\rightarrow\boxminus^{k}\bot)\not\in\KB$.
Let $F=(W,R)$ be a frame, $M=(W,R,V)$ be a model based on $F$, $s\in W$, $F^{\prime}=(W^{\prime},R^{\prime})$ be a frame, $M^{\prime}=(W^{\prime},R^{\prime},V^{\prime})$ be a model based on $F^{\prime}$ and $s^{\prime}\in W^{\prime}$ be such that $M,s\not\models(\sigma(x)\rightarrow\boxplus^{k}\bot)$ and $M^{\prime},s^{\prime}\not\models(\neg\sigma(x)\rightarrow\boxminus^{k}\bot)$.
Hence, $M,s\models\sigma(x)$, $M,s\not\models\boxplus^{k}\bot$, $M^{\prime},s^{\prime}\models\neg\sigma(x)$ and $M^{\prime},s^{\prime}\not\models\boxminus^{k}\bot$.
Let $v_{0},t_{1},u_{1},v_{1},\ldots,t_{k},u_{k},v_{k}\in W$ and $v^{\prime}_{0},t^{\prime}_{1},u^{\prime}_{1},v^{\prime}_{1},\ldots,t^{\prime}_{k},u^{\prime}_{k},v^{\prime}_{k}\in W^{\prime}$ be such that $s=v_{0}$, $s^{\prime}=v^{\prime}_{0}$ and for all $i\in\N$, if $i<k$ then
\begin{itemize}
\item $v_{i}Rt_{i+1}$,
\item $t_{i+1}Ru_{i+1}$,
\item $u_{i+1}Rv_{i+1}$,
\item $v^{\prime}_{i}R^{\prime}t^{\prime}_{i+1}$,
\item $t^{\prime}_{i+1}R^{\prime}u^{\prime}_{i+1}$,
\item $u^{\prime}_{i+1}R^{\prime}v^{\prime}_{i+1}$,
\item $M,v_{i}\models p^{0}\wedge q^{0}$,
\item $M,t_{i+1}\models p^{1}\wedge q^{0}$,
\item $M,u_{i+1}\models p^{0}\wedge q^{1}$,
\item $M,v_{i+1}\models p^{0}\wedge q^{0}$,
\item $M^{\prime},v^{\prime}_{i}\models p^{0}\wedge q^{0}$,
\item $M^{\prime},t^{\prime}_{i+1}\models p^{0}\wedge q^{1}$,
\item $M^{\prime},u^{\prime}_{i+1}\models p^{1}\wedge q^{0}$,
\item $M^{\prime},v^{\prime}_{i+1}\models p^{0}\wedge q^{0}$.
\end{itemize}
Let $M_{s}=(W_{s},R_{s},V_{s})$ be the symmetric unravelling of $M$ around $s$ and $M^{\prime}_{s^{\prime}}=(W^{\prime}_{s^{\prime}},R^{\prime}_{s^{\prime}},V^{\prime}_{s^{\prime}})$ be the symmetric unravelling of $M^{\prime}$ around $s^{\prime}$.
For more on this, see~\cite[Definition~$4.51$]{Blackburn:deRijke:Venema:2001}.
Since $M,s\models\sigma(x)$ and $M^{\prime},s^{\prime}\models\neg\sigma(x)$, therefore by~\cite[Proposition~$2.14$ and Lemma~$4.52$]{Blackburn:deRijke:Venema:2001}, $M_{s},(v_{0})\models\sigma(x)$ and $M^{\prime}_{s^{\prime}},(v^{\prime}_{0})\models\neg\sigma(x)$.
Let $F^{\prime\prime}=(W^{\prime\prime},R^{\prime\prime})$ be the least frame containing the disjoint union of $(W_{s},R_{s})$ and $(W^{\prime}_{s^{\prime}},R^{\prime}_{s^{\prime}})$ and such that for some new states $t$ and $u$,
\begin{itemize}
\item $(v_{0},t_{1},u_{1},v_{1},\ldots,t_{k},u_{k},v_{k})R^{\prime\prime}t$,
\item $tR^{\prime\prime}(v_{0},t_{1},u_{1},v_{1},\ldots,t_{k},u_{k},v_{k})$,
\item $tR^{\prime\prime}t$,
\item $tR^{\prime\prime}u$,
\item $uR^{\prime\prime}t$,
\item $uR^{\prime\prime}u$,
\item $uR^{\prime\prime}(v^{\prime}_{0},t^{\prime}_{1},u^{\prime}_{1},v^{\prime}_{1},\ldots,t^{\prime}_{k},u^{\prime}_{k},v^{\prime}_{k})$,
\item $(v^{\prime}_{0},t^{\prime}_{1},u^{\prime}_{1},v^{\prime}_{1},\ldots,t^{\prime}_{k},u^{\prime}_{k},v^{\prime}_{k})R^{\prime\prime}u$.
\end{itemize}
Let $M^{\prime\prime}=(W^{\prime\prime},R^{\prime\prime},V^{\prime\prime})$ where
\begin{itemize}
\item $V^{\prime\prime}(p)=V_{s}(p)\cup V^{\prime}_{s^{\prime}}(p)\cup\{t\}$,
\item $V^{\prime\prime}(q)=V_{s}(q)\cup V^{\prime}_{s^{\prime}}(q)\cup\{u\}$,
\item for all atoms $\alpha$, if $\alpha\not=p$ and $\alpha\not=q$ then $V^{\prime\prime}(\alpha)=V_{s}(\alpha)\cup V^{\prime}_{s^{\prime}}(\alpha)$.
\end{itemize}
Since $\deg(\sigma(x))\leq3k$, $M_{s},(v_{0})\models\sigma(x)$ and $M^{\prime}_{s^{\prime}},(v^{\prime}_{0})\models\neg\sigma(x)$, therefore $M^{\prime\prime},(v_{0})
$\linebreak$
\models\sigma(x)$ and $M^{\prime\prime},(v^{\prime}_{0})\models\neg\sigma(x)$.
Since $\sigma$ is a unifier of $\varphi$, therefore $((\sigma(x)\rightarrow\boxplus\sigma(x))\wedge(\neg\sigma(x)\rightarrow\boxminus\neg\sigma(x)))\in\KB$.
Since $M^{\prime\prime},(v_{0})\models\sigma(x)$ and $M^{\prime\prime},(v^{\prime}_{0})\models\neg\sigma(x)$, considering that for all $i\in\N$, if $i<k$ then $M,v_{i}\models p^{0}\wedge q^{0}$, $M,t_{i+1}\models p^{1}\wedge q^{0}$, $M,u_{i+1}\models p^{0}\wedge q^{1}$, $M,v_{i+1}\models p^{0}\wedge q^{0}$, $M^{\prime},v^{\prime}_{i}\models p^{0}\wedge q^{0}$, $M^{\prime},t^{\prime}_{i+1}\models p^{0}\wedge q^{1}$, $M^{\prime},u^{\prime}_{i+1}\models p^{1}\wedge q^{0}$ and $M^{\prime},v^{\prime}_{i+1}\models p^{0}\wedge q^{0}$, therefore $M^{\prime\prime},(v_{0},t_{1},u_{1},v_{1},\ldots,t_{k},u_{k},
$\linebreak$
v_{k})\models\sigma(x)$ and $M^{\prime\prime},(v^{\prime}_{0},t^{\prime}_{1},u^{\prime}_{1},v^{\prime}_{1},\ldots,t^{\prime}_{k},u^{\prime}_{k},v^{\prime}_{k})\models\neg\sigma(x)$.
Since $((\sigma(x)\rightarrow\boxplus\sigma(x))\wedge(\neg\sigma(x)\rightarrow\boxminus\neg\sigma(x)))\in\KB$, considering that $M,v_{k}\models p^{0}\wedge q^{0}$, $M^{\prime\prime},t\models p^{1}\wedge q^{0}$, $M^{\prime\prime},u\models p^{0}\wedge q^{1}$ and $M^{\prime},v^{\prime}_{k}\models p^{0}\wedge q^{0}$, therefore $M^{\prime\prime},(v_{0},t_{1},u_{1},v_{1},\ldots,t_{k},
$\linebreak$
u_{k},v_{k})\models\neg\sigma(x)$ and $M^{\prime\prime},(v^{\prime}_{0},t^{\prime}_{1},u^{\prime}_{1},v^{\prime}_{1},\ldots,t^{\prime}_{k},u^{\prime}_{k},v^{\prime}_{k})\models\sigma(x)$.
Thus, $M^{\prime\prime},(v_{0},t_{1},
$\linebreak$
u_{1},v_{1},\ldots,t_{k},u_{k},v_{k})\not\models\sigma(x)$ and $M^{\prime\prime},(v^{\prime}_{0},t^{\prime}_{1},u^{\prime}_{1},v^{\prime}_{1},\ldots,t^{\prime}_{k},u^{\prime}_{k},v^{\prime}_{k})\not\models\neg\sigma(x)$: a contradiction.
\end{proof}
In the proof of Lemma~\ref{lemma:6:K}, remark that the symmetric unravellings $M_{s}$ and $M^{\prime}_{s^{\prime}}$ are serial when the models $M$ and $M^{\prime}$ are serial.
Moreover, when the models $M$ and $M^{\prime}$ are reflexive, $M_{s}$ and $M^{\prime}_{s^{\prime}}$ can be defined as their reflexive symmetric unravellings.
\begin{proposition}\label{lemma:7:K}
$\varphi$ is nullary.
\end{proposition}
\begin{proof}
Suppose $\varphi$ is not nullary.
Let $\Sigma$ be a minimal complete set of unifiers of $\varphi$.
By Lemma~\ref{lemma:0:K}, $\sigma_{0}$ is a unifier of $\varphi$.
Since $\Sigma$ is a complete set of unifiers of $\varphi$, therefore let $\sigma\in\Sigma$ be such that $\sigma\preceq\sigma_{0}$.
Hence, by Lemma~\ref{lemma:6:K}, let $k\in\N$ be such that $\sigma_{k}\preceq\sigma$, or $\tau_{k}\preceq\sigma$.
\\
{\bf ---~Case ``$\sigma_{k}\preceq\sigma$'':}
By Lemma~\ref{lemma:0:K}, $\sigma_{k+1}$ is a unifier of $\varphi$.
Since $\Sigma$ is a complete set of unifiers of $\varphi$, therefore let $\sigma^{\prime}\in\Sigma$ be such that $\sigma^{\prime}\preceq\sigma_{k+1}$.
Since $\sigma_{k}\preceq\sigma$, therefore by Lemma~\ref{lemma:0:K:q}, $\sigma^{\prime}\preceq\sigma$.
Since $\Sigma$ is a minimal set of unifiers of $\varphi$, therefore $\sigma^{\prime}\simeq\sigma$.
Since $\sigma_{k}\preceq\sigma$ and $\sigma^{\prime}\preceq\sigma_{k+1}$, therefore $\sigma_{k}\preceq\sigma_{k+1}$: a contradiction with Lemma~\ref{lemma:0:K:r}.
\\
{\bf ---~Case ``$\tau_{k}\preceq\sigma$'':}
Since $\sigma\preceq\sigma_{0}$, therefore $\tau_{k}\preceq\sigma_{0}$: a contradiction with Lemma~\ref{lemma:0:K:qr}.
\end{proof}
\section{Conclusion}
In modal logic, the problem of checking the unifiability of formulas has been introduced as a special case of the problem of checking the admissibility of inference rules~\cite{Rybakov:1997}.
Intuitively, for an axiomatically presented modal logic, the admissibility problem asks whether a given inference rule can be added to the axiomatization of the logic without changing the associated set of derivable formulas.
Its computability has been studied~---~for a limited number of normal modal logics like $\K{4}$, $\GL$ and $\S{4}$~---~by Je\u{r}\'abek~\cite{Jerabek:2007} and Rybakov~\cite{Rybakov:1984}.
Aside from these transitive normal modal logics and for the normal extensions of $\S{5}$, it is still unknown for numerous normal modal logics~---~for example $\K$, $\KD$ and $\KT$~---~whether the problem of checking the admissibility of inference rules is solvable.
The significance of the unification type in the research on the problem of checking the unifiability of formulas stems from the fact that if a normal modal logic is unitary, or finitary then the problem of checking the admissibility of inference rules can be reduced to the problem of checking the unifiability of formulas.
\\
\\
In this paper, we have adapted to $\KB$ the argument of Je\u{r}\'abek~\cite{Jerabek:2015} showing that $\K$ is nullary, though the nullariness character of $\KB$ have only been be obtained within the context of unification with parameters.
Seeing that the frames constructed in the proofs of Lemmas~\ref{easy:lemma:b} and~\ref{lemma:about:k:l:and:boxes} are reflexive and the symmetric unravellings of the models constructed in the proof of Lemma~\ref{lemma:6:K} are serial when the considered models are serial, or can be forced to be reflexive when the considered models are reflexive, therefore on checking the proofs of our results, the reader may easily verify that our adaptation also applies in the case of $\KDB$ and $\KTB$~---~one has only to replace ``$\KB$'' by ``$\KDB$'', or ``$\KTB$'', ``frame'' by ``serial frame'', or ``reflexive frame'', etc.
The nullariness character of $\KB$, $\KDB$ and $\KTB$ constitutes an answer to questions put forward by Dzik~\cite{Dzik:2007}.
Nevertheless, much remains to be done, seeing that, for instance, the types of simple Church-Rosser normal modal logics like $\KG$, $\KDG$ and $\KTG$ are unknown and for all $k\in\N$ such that $k\geq2$, the type of the least normal modal logic containing $\Box^{k}\bot$ is unknown.
\section*{Acknowledgements}
This paper has been written on the occasion of a $3$-months visit of \c{C}i\u{g}dem Gencer during the Fall $2018$ in Toulouse that was financially supported by {\it Universit\'e Paul Sabatier}\/ (``Professeurs invit\'es 2018'').
We make a point of thanking the colleagues of the {\em Institut de recherche en informatique de Toulouse}\/ who contributed to the development of the work we present today.
Special acknowledgement is also heartily granted to Maryam Rostamigiv (Toulouse University, France) and Tinko Tinchev (Sofia University St. Kliment Ohridski , Bulgaria) for their valuable remarks.
%We finally make a point of thanking the referees for their feedback: their helpful comments and their useful suggestions have been essential for improving the correctness and the readability of a preliminary version of this paper.
%
%
%
%

%
%
%
%
\end{document}